\documentclass[sigconf, dvipsnames]{acmart}
\pdfoutput=1

\usepackage{booktabs} 
\usepackage{courier}
\usepackage{graphicx,graphics}
\usepackage{multirow}
\usepackage{color}
\usepackage{subfig}
\usepackage{balance}
\usepackage{amsmath}
\usepackage{bbm}
\usepackage{algorithmicx}
\usepackage{algorithm}
\usepackage{algpseudocode}
\usepackage{pgfplots}
\usepackage{bm}
\usepackage{cleveref}
\usepackage{physics}

\renewcommand\footnotetextcopyrightpermission[1]{} 

\crefformat{section}{\S#2#1#3} 
\crefformat{subsection}{\S#2#1#3}
\crefformat{subsubsection}{\S#2#1#3}
\newenvironment{customlegend}[1][]{%
	\begingroup
	\csname pgfplots@init@cleared@structures\endcsname
	\pgfplotsset{#1}%
}{%
	\csname pgfplots@createlegend\endcsname
	\endgroup
}%
\def\addlegendimage{\csname pgfplots@addlegendimage\endcsname}

\DeclareMathOperator*{\argmax}{argmax}
\DeclareMathOperator*{\argmin}{argmin}

\newtheorem{theorem}{Theorem}[section]
\newtheorem{lemma}{Lemma}
\newtheorem{definition}{Definition}
\newtheorem{proposition}[theorem]{Proposition}
\newtheorem{claim}{Claim}

\usepackage{enumitem}

\acmConference{Preprint version}{Oct 2020}{(In submission)}

\newcommand{\gourab}[1]{\textcolor{blue}{GP:#1}}

\begin{document}
	
\clubpenalty = 10000
\widowpenalty = 10000

\setlength{\belowdisplayskip}{1pt} 
\setlength{\belowdisplayshortskip}{1pt}
\setlength{\abovedisplayskip}{1pt} 
\setlength{\abovedisplayshortskip}{1pt}

\title{On Fair Virtual Conference Scheduling: \\ Achieving Equitable Participant and Speaker Satisfaction}

\author{Gourab K Patro}
\affiliation{Indian Institute of Technology Kharagpur, India}

\author{Abhijnan Chakraborty}
\affiliation{Max-Planck Institute for Software Systems, Germany}

\author{Niloy Ganguly}
\affiliation{Indian Institute of Technology Kharagpur, India}

\author{Krishna P. Gummadi}
\affiliation{Max-Planck Institute for Software Systems, Germany}

\renewcommand{\shortauthors}{Patro et al.}
	
\begin{abstract}
The (COVID-19) pandemic-induced restrictions on travel and social gatherings have prompted most conference organizers to move their events online. However, in contrast to physical conferences, virtual conferences face a challenge in efficiently scheduling talks, accounting for the availability of participants from different timezones as well as their interests in attending different talks. In such settings, a natural objective for the conference organizers would be to maximize some global welfare measure, such as the total expected audience participation across all talks. However, we show that optimizing for global welfare could result in a schedule that is unfair to the stakeholders, i.e., the individual utilities for participants and speakers can be highly unequal. To address 
the fairness concerns, we formally define fairness notions for participants and speakers, and subsequently derive suitable fairness objectives for them. We show that the welfare and fairness objectives can be in conflict with each other, and  there is a need to maintain a balance between these objective while 
caring for them simultaneously. Thus, we propose a joint optimization framework that allows conference organizers to design talk schedules that balance (i.e., allow trade-offs) between global welfare, participant fairness and the speaker fairness objectives. We show that the optimization problem can be solved using integer linear programming, and 
empirically evaluate the necessity and benefits of such joint optimization approach in virtual conference scheduling.
\end{abstract}
	
\maketitle

\section{Introduction}
Due to the restrictions on travel and social gatherings to tackle the COVID-19 pandemic, most of the conferences have moved online and 
some of them may remain online in the years to come. 
Online conferences are not only hugely economical---due to the reduction in organization and travel costs, thus enabling participation from resource/budget constrained regions---they are also more environmentally sustainable than their physical in-person counterparts (by reducing the carbon footprint from long-distance air travels, huge power and non-renewable consumption at event venues, etc.). 
In addition, they also provide 
a unique opportunity to significantly improve the scale and outreach of the conferences, along with 
focused discussions through interactive tools like live messaging \cite{oxford2020conf}.
However, online conferences come with their own set of challenges. For example, scaling up 
participation is subject to the availability of stable and high-speed Internet in different regions;
time consuming training of participants and speakers on different interactive conferencing tools is essential for efficient participation \cite{saliba2020getting}.
Another big challenge in organizing online conferences is {\it optimal scheduling of the conference talks};
this is because, online conferences usually have participants from different timezones all around the globe unlike the physical conference setup where participants assemble at a single place to participate in the conference.
Thus, traditional timezone-specific conference schedules---usually followed in physical conferences based on the timezone of the venue--- is  no longer convenient for conferences being held online, as the participants from the other distant parts of the globe will find it hard to attend.
This demands for conference schedules which are not timezone-specific but timezone-aware, and may stretch beyond the usual 7-8 hours of a day, 
to cater to the participants from different timezones.
In this paper, we focus on this conference scheduling problem and relevant concerns of efficiency and fairness.
A natural objective for conference scheduling would be to maximize some social welfare measure like the total expected audience participation across all talks (formally defined in \cref{subsec:tep}).
However, optimizing for such a social welfare objective could result in a schedule that is unfair to the stakeholders (as illustrated in \cref{subsec:tensions}), i.e., the level of satisfaction enjoyed by individual participants (formally defined in \cref{subsec:participant_satis}) can be very different, and the expected exposure 
(audience size) at different talks can be disproportionately 
skewed ---~leading to disparity in speaker satisfactions (formally defined in \cref{subsec:speaker_satis}).
Intuitively, a participant would be less satisfied if her favorite talks are scheduled in timeslots that are unfavorable for her, and similarly a speaker would be less satisfied if her talk is scheduled in a timeslot which adversely affects her deserved exposure (expected audience or crowd).
Thus, in conference scheduling, fairness for participants and speakers are also desired along with social welfare.

We formally define the problem setup alongside suitable measures of participant satisfaction, speaker satisfaction, and social welfare in \cref{sec:prelim}. 
Intuitively, stakeholder fairness would be in bringing parity to their normalized satisfactions (parity in individual participant satisfactions and parity in individual speaker satisfactions).
However, as absolute parity is often hard to achieve in discrete real-world problems, we propose suitable relaxed fairness notions for participants and speakers in \cref{subsec:p_fairness,subsec:s_fairness} respectively.
Subsequently, we reduce these fairness notions to fairness objectives through suitable unfairness measures.
Both fairness objectives along with welfare objective are important in conference scheduling.
However, it may  be impossible to optimize all three simultaneously as there are some fundamental tensions among these objectives (more details in \cref{subsec:tensions});
optimizing for one objective could cause losses in other objectives.

We propose a joint optimization framework (\cref{subsec:joint_opt}) that allows conference organizers to design schedules that balance (i.e., allow trade-offs) between the global welfare, the participant fairness and the speaker fairness objectives.
We show that the joint optimization problem can be solved using integer linear programming, and we also empirically evaluate the necessity and benefits of such joint optimization approach in virtual conference scheduling along with the analysis on the pitfalls of baseline approaches (\cref{sec:experiments}).
Our focus, in this paper, is more towards bringing out the fundamental tensions, trade-offs, and difficulties involved in online conference scheduling.
With this work, we begin to lay the foundation towards more focused research into---very timely and important problem---efficient and fair online conference scheduling, and hope to motivate a new line of work---both of theoretical and empirical interests---on such multi-stakeholder scheduling settings.
In this regard, we also provide a detailed description on possible future works (\cref{sec:discussion}) to consider many new nuances observed in online conferences.
In summary, we make the following contributions.
\begin{itemize}
	\item We formally define the problem of online conference scheduling (\cref{sec:prelim}), and the notions of social welfare (\cref{subsec:tep}), participant fairness (\cref{subsec:p_fairness}) and speaker fairness (\cref{subsec:s_fairness}). 
	Through suitable unfairness measures, we reduce them to fairness objectives.
	To our knowledge, we are the first to do so.
	\item We illustrate the some fundamental tensions, trade-offs, and possibility of conflicts among the two fairness objectives and welfare objective (\cref{subsec:tensions}).
	\item We propose a joint optimization problem to suitably balance these objectives, and empirically illustrate the difficulties involved in the problem and the benefits of our approach (\cref{subsec:joint_opt,sec:experiments}). We also detail other nuances involved in online conference and possible future works (\cref{sec:discussion}).
\end{itemize}

\section{Preliminaries}\label{sec:prelim}
{\bf Problem Setup:}
In a conference, let $\mathcal{P}$, $\mathcal{T}$, and $\mathcal{S}$ represent the sets of participants, planned talks, and available slots (non-overlapping) respectively;
$|\mathcal{P}|=m$, $|\mathcal{T}|=n$, $|\mathcal{S}|=l$;
let $p\in \mathcal{P}$, $t\in\mathcal{T}$, and $s\in\mathcal{S}$ be instances of participant, talk, and slot respectively.
Assuming a talk to be scheduled only once, a conference schedule $\Gamma$ is a mapping $\Gamma:\mathcal{T}\rightarrow\mathcal{S}$.
Note that, in this paper, we limit ourselves to the case with no parallel or overlapping time slots;
this implies that each slot refers to a unique time interval.
Thus, the conference schedule $\Gamma$ is a one-to-one mapping with $n\leq l$.
The goal of a conference scheduling problem is to find a schedule $\Gamma$ which satisfies some specified constraint(s) or optimizes some specified objective(s).

~\\\noindent {\bf Interest Scores $[V_p(\cdot)]$:}
The participants may have different preference levels over the set of talks.
We model this phenomenon using participant-specific interest scores.
Let $V(t|p)=V_p(t)$ represent $p$'s interest score for talk $t$.
Note that the interest score represents the probability of satisfaction of the participant on attending the corresponding talk;
i.e., $V_p(t)\in [0,1], \forall p\in\mathcal{P},t\in\mathcal{T}$.

~\\\noindent {\bf Ease of Availability $[A_p(\cdot)]$:}
In a virtual conference setting, the participants are located in different parts of the world which makes it convenient for them to attend talks only in specific times of the day 
(usually during the day time of their timezone).
Note that participants belonging to same timezone can still have different ease of availability throughout the $24$-hour period.
We model this phenomenon using participant-specific availability scores.
Let $A(s|p)=A_p(s)$ represent the ease of availability score or the probability of $p$ making herself available in slot $s$;
i.e., $A_p(s)\in [0,1], \forall p\in\mathcal{P},s\in\mathcal{S}$. 

\subsection{Participant Satisfaction ($NCG$)}\label{subsec:participant_satis}
In virtual conference settings, a participant's satisfaction depends on both her interest for the talks and her ease of availability for the time slots when the talks are scheduled.
For simplicity, we assume $V_p(\cdot)$ and $A_p(\cdot)$ to be independent of each other, i.e., the interest score of a talk does not affect the ease of availability of a participant in a slot and vice-versa.
However, a participant may still attend a talk scheduled in a slot with less availability score if she has 
a very high interest score for the talk.
This means that the expected gain of a participant $p$ from a talk $t$ in slot $s$ will depend on the joint probability of $p$ making herself available in $s$ and she getting satisfied after attending $t$: i.e, $V_p(t)\times A_p(s)$.
Note that here $V_p(t)\times A_p(s)$ also represents the probability of the participant $p$ attending talk $t$ in slot $s$.
Thus, given a conference schedule $\Gamma$, we define the cumulative gain ($CG$) of a participant as below.
\begin{equation}
	CG(p|\Gamma)=CG_p(\Gamma)=\sum_{t\in\mathcal{T}} V_p(t)\times A_p(\Gamma(t))
\end{equation}
Now, let's imagine a situation wherein the participant $p$ is asked to choose the conference schedule $\Gamma$.
Assuming $p$ to be a selfish and rational agent, she would choose the schedule which benefits her the most; i.e., the one which gives her the highest cumulative gain.
Here, the best conference schedule for $p$ would be the one in which the talk with the highest $V_p$ is scheduled in the slot with the highest $A_p$, the talk with second highest $V_p$ is scheduled in the slot with the second highest $A_p$, and so on;
let $\Gamma_p^*$ be that best conference schedule for $p$.
We call the cumulative gain of $p$ given $\Gamma_p^*$ as her ideal cumulative gain ($ICG$) as defined below.
\begin{equation}
	ICG(p)=ICG_p=\max_\Gamma CG_p(\Gamma)= CG_p(\Gamma_p^*)
\end{equation}
$ICG$ represents the maximum possible cumulative gain depending on the interest scores of the participant and her availability scores.
Thus, a participant with higher overall interests or higher overall availability will naturally have higher $ICG$.
We now define the overall satisfaction of a participant as her normalized cumulative gain ($NCG$) as below.
\begin{equation}
	NCG(p|\Gamma)=NCG_p(\Gamma)=\dfrac{CG_p(\Gamma)}{ICG_p}
\end{equation}
%
The denominator $ICG_p$ is the maximum possible cumulative gain for the participant $p$.
Thus, $NCG_p(\Gamma)\in [0,1]$, $\forall p\in \mathcal{P}$, $\forall \Gamma$.
\subsection{Speaker Satisfaction ($NEC$)}\label{subsec:speaker_satis}
The speaker of a talk gets satisfied if her talk gets participation;
more the participation more will be the speaker satisfaction.
To have higher participation, the talk needs to be scheduled in a slot with high ease of availability of the participants who are highly interested in the talk.
Thus, given a schedule $\Gamma$, we define expected crowd ($EC$) at talk $t$ as below.
\begin{equation}
	EC(t|\Gamma)=EC_t(\Gamma)=\sum_{p\in\mathcal{P}}V_p(t)\times A_p(\Gamma(t))
\end{equation}
Now if the speaker of the talk $t$ is given the task to design the conference schedule, she would try to maximize her expected crowd assuming that each speaker is a selfish and rational agent;
let that best schedule for $t$ be denoted as $\Gamma_t^*$.
We call the expected crowd at talk $t$ with schedule $\Gamma_t^*$ as the ideal expected crowd ($IEC$) of $t$. 
\begin{equation}
	IEC(t)=IEC_t=\max_\Gamma EC_t(\Gamma)=EC_t(\Gamma_t^*)
\end{equation}
$IEC$ represents the maximum value for expected crowd depending on the overall interest scores of the participants and their availability scores.
Thus, a talk with higher overall interests from the participants will have higher $IEC$.
We now define the overall satisfaction of a speaker as the normalized expected crowd ($NEC$) at her talk as below.
\begin{equation}
	NEC(t|\Gamma)=NEC_t(\Gamma)=\dfrac{EC_t(\Gamma)}{IEC_t}
\end{equation}
Note that the denominator $IEC_t$ is the maximum possible value of the expected crowd at talk $t$.
Thus, $NEC_t(\Gamma)\in [0,1]$, $\forall t\in \mathcal{T}$, $\forall \Gamma$.
\subsection{Social Welfare Objective ($TEP$)}\label{subsec:tep}
A natural objective of the  conference organizers is maximizing the social welfare,
i.e., the total participation in the conference.
Given a talk $t$ scheduled in slot $s$, the expectation of the participant $p$ attending it, can be written as $V_p(t)\times A_p(s)$.
Thus, the total expected participation ($TEP$) in the conference with schedule $\Gamma$ can be written as below.
\begin{equation}
	TEP(\Gamma)=  \sum_{p\in \mathcal{P}} \sum_{t\in \mathcal{T}} V_p(t)\times A_p(\Gamma(t))
	\label{eq:TEP}
\end{equation}
It is worth noting that the social welfare, here, is same as: (i) {\it the sum of cumulative gains of all the participants}; (ii) {\it the sum of expected crowd at all the talks}; i.e.,
\begin{multline}
	TEP(\Gamma)= \sum_{p\in \mathcal{P}}\sum_{t\in \mathcal{T}} V_p(t)\times A_p(\Gamma(t))=\sum_{p\in \mathcal{P}} CG(p|\Gamma)\\
	= \sum_{t\in \mathcal{T}}\sum_{p\in \mathcal{P}} V_p(t)\times A_p(\Gamma(t))=\sum_{t\in \mathcal{T}} EC(t|\Gamma)
\end{multline}

Now the natural social welfare objective for conference scheduling can be written as $\argmax_\Gamma TEP(\Gamma)$.
We use $\Gamma^\text{SW}$ to represent the schedule which maximizes the social welfare;
i.e., $\Gamma^\text{SW}=\argmax_\Gamma TEP(\Gamma)$.

\section{Fairness in Conference Scheduling}\label{sec:motivation}
Optimizing for just participation-based objective could result in the schedule being unfair to the stakeholders involved, i.e., participants and speakers.
Optimum participation could result in very high satisfaction for some participants while other participants may end up very less satisfied;
same could happen for the speakers too.
Thus, we first define the fairness notions for the participants and the speakers in \cref{subsec:p_fairness} and \cref{subsec:s_fairness} respectively.
\subsection{Participant Fairness}\label{subsec:p_fairness}
Disparity in normalized satisfactions of participants is the cause of participant unfairness.
Thus, to ensure fairness for participants, the conference schedule should equally satisfy all the participants.
However, such hard constraint can be infeasible in real-world cases.
Thus, we define a relaxed fairness notion for participants below.
\begin{definition}
	{\bf \boldmath$\epsilon$-Fairness for Participants:}
	For a non-negative real number $\epsilon$, a schedule $\Gamma$ is said to be $\epsilon$-fair for participants iff the following condition is satisfied.
	\begin{equation}
	\abs{NCG(p_i|\Gamma)-NCG(p_j|\Gamma)} \leq \epsilon, \space \forall p_i,p_j\in \mathcal{P}
	\label{eq:participant_fairness}
	\end{equation}
	\label{def:participant_fairness}
\end{definition}
Note that, smaller the value of $\epsilon$, lesser is the disparity between participant satisfactions, and fairer is the conference schedule for the participants.
Thus, to relate participant fairness notions with different $\epsilon$ values, we can write the following lemma.
\begin{lemma}
	If a schedule $\Gamma$ is $\epsilon'$-fair for the participants, then it is also $\epsilon$-fair for participants $\forall \epsilon\geq \epsilon'$.
	\label{lemma:participant_fairness}
\end{lemma}
\begin{proof}
	If pairwise disparities in participant satisfaction are less than $\epsilon'$, then they are also less than $\epsilon$, as $\epsilon'\leq \epsilon$.
\end{proof}

In our setting, the value of $\epsilon$ can also be thought of as the tolerance level for disparity or unfairness in participant satisfactions.
Thus, given a schedule $\Gamma$, we can find the smallest possible $\epsilon$ and use that to represent how unfair is $\Gamma$ to the participants. 
We formally define the metric to measure participant unfairness below.
\begin{definition}
	{\bf Participant Unfairness $\big(\Psi^\text{P}(\Gamma)\big)$:}
	The participant unfairness caused by a schedule $\Gamma$, is the smallest non-negative value of $\epsilon$ such that $\Gamma$ is $\epsilon$-fair for the participants.
	\begin{equation}
		\Psi^\text{P}(\Gamma)=\inf \big\{\epsilon \colon \epsilon\geq 0 \text{ }\& \text{ } \Gamma \text{ }\text{\em  is } \epsilon\text{\em -fair for participants} \big\}
		\label{eq:participant_unfairness}
	\end{equation}
	\label{def:participant_unfairness}
\end{definition}
(Here, the notation $\inf\{\cdot\}$ represents infimum of a set.)
\begin{proposition}
	Participant unfairness metric from \cref{eq:participant_unfairness} can be reduced as below.
	\begin{equation}
		\Psi^\text{P}(\Gamma)=\big\{\max_{p_i\in\mathcal{P}} NCG(p_i|\Gamma)\big\}-\big\{\min_{p_j\in\mathcal{P}} NCG(p_j|\Gamma)\big\}
		\label{eq:participant_unfairness_redn}
	\end{equation}
	\label{prop:participant_unfairness_redn}
\end{proposition}
\begin{proof}
	For a schedule $\Gamma$, let's assume $\Psi^\text{P}(\Gamma)=\epsilon'$.
	Let $\epsilon\in [0,\epsilon')$, then $\Gamma$ is not $\epsilon$-fair for participants (from \cref{def:participant_unfairness}).
	%
	This implies that, there is a pair of participants $p_i,p_j\in\mathcal{P}$ such that $\epsilon<\abs{NCG(p_i|\Gamma)-NCG(p_j|\Gamma)}\leq\epsilon'$ (from \cref{def:participant_fairness}).
	Let's increase $\epsilon$ till $\epsilon=\big|NCG(p_i|\Gamma)-NCG(p_j|\Gamma)\big|$.
	Thus, now the satisfaction disparity between $p_i$ and $p_j$ is at most $\epsilon$, however $\epsilon$ may still be less than $\epsilon'$.
	If we keep on increasing $\epsilon$ like this, there will not be any such opportunity to increase $\epsilon$ further when we reach $\epsilon=\max_{p_i,p_j\in\mathcal{P}} \abs{NCG(p_i|\Gamma)-NCG(p_j|\Gamma)}$.
	Now with the current value of $\epsilon$, $\Gamma$ can be called $\epsilon$-fair for the first time.
	Therefore, $\Psi^\text{P}(\Gamma)=\epsilon'=\max_{p_i,p_j\in\mathcal{P}} \abs{NCG(p_i|\Gamma)-NCG(p_j|\Gamma)}$, which is the maximum pairwise disparity in participant satisfaction.

	As $\forall p\in \mathcal{P}$,
	\[\big\{\max_{p_i\in\mathcal{P}} NCG(p_i|\Gamma)\big\}\geq NCG(p|\Gamma)\geq \big\{\min_{p_j\in\mathcal{P}} NCG(p_j|\Gamma)\big\}\geq0,\]
	the maximum pairwise disparity will be in between the most satisfied participant(s) and the least satisfied participant(s).\\
	$\therefore \Psi^\text{P}(\Gamma)=\big\{\max_{p_i\in\mathcal{P}} NCG(p_i|\Gamma)\big\}-\big\{\min_{p_j\in\mathcal{P}} NCG(p_j|\Gamma)\big\}$
\end{proof}
Using the metric for unfairness above, we can define the following fairness objective for participants.
\begin{definition}
	{\bf Fairness objective for participants} can be defined as finding the schedule $\Gamma$ which minimizes $\Psi^\text{P}(\Gamma)$.
	\begin{equation}
		\argmin_\Gamma \Psi^\text{P}(\Gamma)\equiv \argmin_\Gamma \bigg\{\Big\{\max_{p_i\in\mathcal{P}} NCG(p_i|\Gamma)\Big\}-\Big\{\min_{p_j\in\mathcal{P}} NCG(p_j|\Gamma)\Big\}\bigg\}
		\label{eq:participant_fairness_obj}
	\end{equation}
	\label{def:participant_fairness_obj}
\end{definition}
\subsection{Speaker Fairness}\label{subsec:s_fairness}
Similar to participant fairness, we follow a parity-based notion for speaker fairness and define the relaxed speaker fairness below.
\begin{definition}
	{\bf \boldmath$\epsilon$-Fairness for Speakers}:
	For a non-negative real number $\epsilon$, a schedule $\Gamma$ is said to be $\epsilon$-fair for speakers iff the following condition is satisfied.
	\begin{equation}
	\abs{NEC(t_i|\Gamma)-NEC(t_j|\Gamma)} \leq \epsilon, \space \forall t_i,t_j\in \mathcal{T}
	\label{eq:speaker_fairness}
	\end{equation}
	\label{def:speaker_fairness}
\end{definition}
The \cref{lemma:participant_fairness} can be repeated, here, for speakers too.
\begin{lemma}
	If a schedule $\Gamma$ is $\epsilon'$-fair for the speakers, then it is also $\epsilon$-fair for speakers $\forall \epsilon\geq \epsilon'$.
	\label{lemma:speaker_fairness}
\end{lemma}
Following arguments similar to participant unfairness, we define the metric to measure speaker unfairness below.
\begin{definition}
	{\bf Speaker Unfairness $\big(\Psi^\text{S}(\Gamma)\big)$:}
	The speaker unfairness caused by a schedule $\Gamma$, is the smallest non-negative value of $\epsilon$ such that $\Gamma$ is $\epsilon$-fair for the speakers.
	\begin{equation}
	\Psi^\text{S}(\Gamma)=\inf \big\{\epsilon \colon \epsilon\geq 0 \text{ }\& \text{ } \Gamma \text{ }\text{\em  is } \epsilon\text{\em -fair for speakers} \big\}
	\label{eq:speaker_unfairness}
	\end{equation}
	\label{def:speaker_unfairness}
\end{definition}
\begin{proposition}
	Participant unfairness metric from \cref{eq:participant_unfairness} can be reduced as below.
	\begin{equation}
	\Psi^\text{S}(\Gamma)=\big\{\max_{t_i\in\mathcal{T}} NEC(t_i|\Gamma)\big\}-\big\{\min_{t_j\in\mathcal{T}} NEC(t_j|\Gamma)\big\}
	\label{eq:speaker_unfairness_redn}
	\end{equation}
	\label{prop:speaker_unfairness_redn}
\end{proposition}
We skip the proofs for \cref{lemma:speaker_fairness} and \cref{prop:speaker_unfairness_redn}, as they will follow arguments and approach similar to those in \cref{lemma:participant_fairness} and \cref{prop:participant_unfairness_redn} respectively.
Using the speaker unfairness metric from \cref{eq:speaker_unfairness_redn}, we define the following fairness objective for speakers.
\begin{definition}
	{\bf Fairness objective for speakers} can be defined as finding the schedule $\Gamma$ which minimizes $\Psi^\text{S}(\Gamma)$.
	\begin{equation}
	\argmin_\Gamma \Psi^\text{S}(\Gamma)\equiv \argmin_\Gamma \bigg\{\Big\{\max_{t_i\in\mathcal{T}} NEC(t_i|\Gamma)\Big\}-\Big\{\min_{t_j\in\mathcal{T}} NEC(t_j|\Gamma)\Big\}\bigg\}
	\label{eq:speaker_fairness_obj}
	\end{equation}
	\label{def:speaker_fairness_obj}
\end{definition}
\section{Balancing Welfare and Fairness}
Given the conference scheduling problem (as defined in \cref{sec:prelim}), the ultimate goal is to find a schedule $\Gamma$ which optimizes social welfare while minimizing participant unfairness and speaker unfairness (as defined in \cref{eq:TEP}, \cref{eq:participant_fairness_obj}, \cref{eq:speaker_fairness_obj} respectively).
%
%
However, simultaneously optimizing for them could prove to be challenging and present disagreement between the objectives.
Thus, first in \cref{subsec:tensions}, we bring out certain difficulties in simultaneously ensuring fairness and social welfare while also highlighting the potential conflicts between them.
Then, in \cref{subsec:joint_opt}, we propose a joint optimization framework for the problem.
\subsection{Tension between Fairness and Welfare}\label{subsec:tensions}
In this section, through a set of claims, 
we illustrate some fundamental tensions between social welfare and fairness.
\begin{claim}
	In the virtual conference scheduling problem (as defined in \cref{sec:prelim}), it is not always possible to gain participant fairness without losing social welfare.
	\label{claim:1}
\end{claim}
\begin{proof}
	Let's assume the opposite argument to be true; i.e., we can always gain participant fairness with no loss to social welfare.
	We disprove this using a counter example with two participants, one talk and three slots as given in \cref{tab:toy_example_1}.
	Both participants have interest score of $1$ for the talk.
	Now looking at the availability scores, while $p_1$ and $p_2$ have full ease of availability in $s_1$ and $s_3$ respectively, they both can make themselves available $s_2$ with $0.49$ probability.
	If we consider a social welfare objective (participation maximization) here, we would end up scheduling the talk in either in $s_1$ or in $s_3$;
	if $\Gamma(t)=s_1$ or $\Gamma(t)=s_3$, then $TEP(\Gamma)=1$;
	if $\Gamma(t)=s_2$, then $TEP(\Gamma)=0.98$ which is less.
	However, maximizing participation will either end up with [$NCG(p_1)=1,NCG(p_2)=0$] if $\Gamma(t)=s_1$ or [$NCG(p_1)=0,NCG(p_2)=1$] if $\Gamma(t)=s_3$;
	As both of these results from social welfare optimization provide disparate satisfaction to the participants, they both are unfair.
	On the other hand, if we schedule the talk in $s_2$ ($\Gamma(t)=s_2$), then it becomes fair to both the participants as they will get similar satisfaction [$NCG(p_1)=0.49,NCG(p_2)=0.49$]---they both get a chance to make themselves available in $s_2$ to attend the talk.
	Even though scheduling the talk in $s_2$, ensures participant fairness, it has come at a loss in social welfare; i.e., $TEP(\Gamma)$ reduced from $1$ to $0.98$.
\end{proof}

\begin{table}[t]
\begin{tabular}{*{5}{|c|c|c|c|c}p{5em}}
	\hline
	\multirow{2}*{\bf Participants} & \multicolumn{1}{c|}{\bf $V_p(t)$} & \multicolumn{3}{c|}{\bf $A_p(s)$}\\\cline{2-5}
	& $t$ & $s_1$ & $s_2$& $s_3$\\\cline{1-5}
	$p_1$ & $1$ & $1$& $0.49$ &$0$ \\\cline{1-5}
	$p_2$  & $1$ &$0$ & $0.49$&$1$\\\cline{1-5}   
\end{tabular}
\caption{\bf Example Problem 1}
\label{tab:toy_example_1}
\vspace{-4mm}
\end{table}
Note that, in the example given in \cref{tab:toy_example_1}, introducing participant fairness has caused a loss in social welfare, and also a loss in the speaker satisfaction $NEC(t)$ (also reduced from $1$ to $0.98$).
Thus, it is important to ask: \textbf{Q1: \textit{to what extent the conference organizer is ready to lose social welfare and speaker satisfaction while bringing in participant fairness?}}
%
\begin{claim}
	In the virtual conference scheduling problem (as defined in \cref{sec:prelim}), it is not always possible to gain speaker fairness with out losing social welfare.
	\label{claim:2}
\end{claim}
\begin{proof}
	Let's assume the opposite argument to be true; i.e., we can always gain speaker fairness with no loss to social welfare.
	We disprove this using a counter example with just one participant, two talks, and three available slots as given in \cref{tab:toy_example_2}.
	Now, to maximize social welfare, we can just match talks in decreasing order of overall interest scores to slots in decreasing order of availability scores; i.e., $\Gamma^\text{SW}(t_1)=s_1$ and $\Gamma^\text{SW}(t_2)=s_3$, which will yield $TEP(\Gamma^\text{SW})=1.4$.
	The speaker satisfactions for the talks with this schedule will be:
	$NEC(t_1|\Gamma^\text{SW})=1$ (as $EC(t_1|\Gamma^\text{SW})=1$ and $IEC(t_1)=1$), and $NEC(t_2|\Gamma^\text{SW})=0.8$ (as $EC(t_1|\Gamma^\text{SW})=0.4$ and $IEC(t_1)=0.5$).
	Such disparity in speaker satisfactions can be attributed to speaker unfairness.
	In order to reduce speaker-side disparity, we can use a different schedule:
	$\Gamma(t_1)=s_3$ and $\Gamma(t_2)=s_2$;
	this yields speaker satisfactions $NEC(t_1|\Gamma)=0.8$ and $NEC(t_2|\Gamma)=0.75$ (as $EC(t_1|\Gamma)=0.8$ and $EC(t_2|\Gamma)=0.375$).
	This above schedule has in fact the lowest possible disparity in speaker satisfactions, i.e., the highest possible speaker fairness.
	Even though this schedule $\{(t_1,s_3),(t_2,s_2)\}$ is fairer to the speakers than the earlier $\{(t_1,s_1),(t_2,s_3)\}$, the gain in speaker fairness has come at a loss in social welfare; $TEP(\Gamma)$ reduced from $1.4$ to $1.175$.
\end{proof}
\begin{table}[b]
\begin{tabular}{*{6}{|c|c|c|c|c|c}p{5em}}
	\hline
	\multirow{2}*{\bf Participants} & \multicolumn{2}{c|}{\bf $V_p(t)$} & \multicolumn{3}{c|}{\bf $A_p(s)$}\\\cline{2-6}
	& $t_1$ & $t_2$ & $s_1$ & $s_2$& $s_3$\\\cline{1-6}
	$p$ & $1$&$0.5$ &$1$ &$0.75$ &$0.8$ \\\cline{1-6}
\end{tabular}
\caption{\bf Example Problem 2}
\vspace{-4mm}
\label{tab:toy_example_2}
\end{table}

There are two important points to note from the example in \cref{tab:toy_example_2}:
(i) the most fair solution for speakers leaves a very valuable slot $s_1$---with the highest overall availability scores for participants---unused thereby losing a huge opportunity for larger participation;
(ii) speaker fairness has introduced a loss in social welfare ($TEP(\Gamma^\text{SW})=1.4$ to $TEP(\Gamma)=1.175$) and also a loss in participant satisfaction ($NCG(p|\Gamma^\text{SW})=1$ to $NCG(p|\Gamma)\approx0.84$).
Thus, it is important to ask; \textbf{Q2: \textit{to what extent the conference organizer is ready to lose social welfare and participant satisfaction while bringing in speaker fairness?}}
\begin{claim}
	In the virtual conference scheduling problem (as defined in \cref{sec:prelim}), it is not always possible to get speaker fairness with out losing participant fairness and vice-versa.
	\label{claim:3}
\end{claim}
\begin{proof}
	Let's assume the opposite argument to be true;
	we can always get participant fairness and speaker fairness simultaneously.
	We disprove this using a counter example with two participants, two talks and four available slots as given in \cref{tab:toy_example_3}.
	In this example, the schedule $\Gamma=\{(t_1,s_2),(t_2,s_3)\}$ achieves speaker fairness---$NEC(t_1|\Gamma)=NEC(t_2|\Gamma)=0.5$ (as $EC(t_1|\Gamma)=1$, $EC(t_2|\Gamma)=0.7$, while $IEC(t_1)=2$, $IEC(t_2)=1.4$).
	However, $\Gamma$ is unfair for the participants---$NCG(p_1|\Gamma)=\frac{1}{1.7}<\frac{0.7}{1.7}=NCG(p_2|\Gamma)$ (as $CG(p_1|\Gamma)=1$, $CG(p_2|\Gamma)=0.7$, while $ICG(p_1)=ICG(p_2)=1.7$).
	On the other hand, schedule $\Gamma'=\{(t_1,s_1),(t_2,s_4)\}$ is fair for the participants---$NCG(p_1|\Gamma')=NCG(p_2|\Gamma')=\frac{1.14}{1.7}$, while being unfair for the speakers as $NEC(t_1|\Gamma')=1>0.2=NEC(t_2|\Gamma')$.
	\begin{table}[h]
		\begin{tabular}{*{7}{|c|c|c|c|c}p{5em}}
			\hline
			\multirow{2}*{\bf Participants} & \multicolumn{2}{c|}{\bf $V_p(t)$} & \multicolumn{4}{c|}{\bf $A_p(s)$}\\\cline{2-7}
			& $t_1$ & $t_2$ & $s_1$ & $s_2$& $s_3$ & $s_4$\\\cline{1-7}
			$p_1$ & $1$ & $0.7$ & $1$ & $1$ & $0$ & $0.2$ \\\cline{1-7}
			$p_2$ & $1$ & $0.7$ & $1$ & $0$ & $1$ & $0.2$ \\\cline{1-7}   
		\end{tabular}
		\caption{\bf Example Problem 3}\label{tab:toy_example_3}
	\end{table}
\end{proof}
As seen in the example in \cref{tab:toy_example_3}, participant fairness and speaker fairness can not always be achieved simultaneous.
Thus, an important question is: \textbf{\textit{Q3: to what extent the conference organizer is ready to lose participant fairness while bringing in speaker fairness and vice-versa?}}

~\\It is very evident from the given examples that, both fairness and satisfaction of participants and speakers could often come in conflict with each other and also with the overall social welfare in such conference scheduling problems.
Thus, there is a need to maintain suitable balance between both fairness notions and social welfare while still simultaneously caring for them in conference scheduling.
\subsection{Joint Optimization for Welfare and Fairness}\label{subsec:joint_opt}
We combine participant and speaker fairness objectives with our natural objective of social welfare maximization, and design the following joint optimization problem.
\begin{multline}
\argmax_\Gamma \frac{TEP(\Gamma)}{mn}+ \lambda_1\times \bigg\{\Big\{\min_{p_j\in\mathcal{P}} NCG(p_j|\Gamma)\Big\}-\max_{p_i\in\mathcal{P}} NCG(p_i|\Gamma)\Big\}\bigg\}\\
+\lambda_2\times \bigg\{\Big\{\min_{t_j\in\mathcal{T}} NEC(t_j|\Gamma)\Big\}-\Big\{\max_{t_i\in\mathcal{T}} NEC(t_i|\Gamma)\Big\}\bigg\}
\label{eq:joint_opt}
\end{multline}
Here we normalize the social welfare objective to bring all the three components to similar scales;
i.e., $TEP$ is divided by $|\mathcal{P}|\cdot|\mathcal{T}|=mn$ (it is the maximum possible value for $TEP$---occurs when $V_p(t)=A_p(s)=1$, $\forall p,t,s$).
We also reverse the fairness objective functions from \cref{eq:participant_fairness_obj} and \cref{eq:speaker_fairness_obj} while inserting them in \cref{eq:joint_opt} as it features $\argmax$ instead of $\argmin$, and use $\lambda_1,\lambda_2$ as weights for participant fairness and speaker fairness respectively.

We take a matrix $X$ of dimensions $|\mathcal{T}|\times|\mathcal{S}|$.
The elements of $X$: $X_{t,s}$ is a binary indicator variable for talk $t\in \mathcal{T}$ being scheduled in slot $s\in \mathcal{S}$, i.e., $X_{t,s}=1$ if $t$ is scheduled in $s$ and $0$ otherwise.
Now to operationalize the joint optimization objective in \cref{eq:joint_opt}, we express it as the following integer linear program.

\begin{multline}\small
\argmax_X \text{   }\space \frac{1}{mn}\sum\limits_{t\in\mathcal{T}}\sum\limits_{p\in\mathcal{P}} \sum\limits_{s\in\mathcal{S}}V_p(t)\cdot A_p(s)\cdot X_{t,s}\\ 
+\lambda_1\Bigg[\min_{p_j\in \mathcal{P}} \sum_{t\in\mathcal{T}}\sum_{s\in\mathcal{S}}\frac{V_{p_j}(t)\cdot A_{p_j}(s)}{ICG(p_j)}\cdot X_{t,s}\\-\max_{p_i\in \mathcal{P}}\sum_{t\in\mathcal{T}}\sum_{s\in\mathcal{S}}\frac{V_{p_i}(t)\cdot A_{p_i}(s)}{ICG(p_i)}\cdot X_{t,s}\Bigg]\\
+\lambda_2\Bigg[\min_{t_j\in \mathcal{T}} \sum_{p\in\mathcal{P}}\sum_{s\in\mathcal{S}}\frac{V_{p}(t_j)\cdot A_{p}(s)}{IEC(t_j)}\cdot X_{t_j,s}\\-\max_{t_i\in \mathcal{T}}\sum_{p\in\mathcal{P}}\sum_{s\in\mathcal{S}}\frac{V_{p}(t_i)\cdot A_{p}(s)}{IEC(t_i)}\cdot X_{t_i,s}\Bigg]\\
\text{s.t.   }\\
X_{t,s}\in \{0,1\} \text{   }\forall t\in\mathcal{T}, s\in\mathcal{S}\\
\sum_{s\in \mathcal{S}}X_{t,s}=1,\text{   }\forall t\in\mathcal{T}\\
\sum_{t\in \mathcal{T}}X_{t,s}\leq1,\text{   }\forall s\in\mathcal{S}\\
\end{multline}
Here the first constraint is the integral constraint for the interger program.
Second constraint ensures that, each talk gets scheduled once.
On the other hand, one slot can be allocated to atmost one talk which is ensured by the third constraint.
We use cvxpy (\url{https://www.cvxpy.org/}) paired with Gurobi (\url{https://www.gurobi.com/}) solver for the integer linear programs.\footnote{Even though we solve the joint optimization problem by converting it into an integer linear program which works well for problems of small size, more research is needed in developing approaches 
that are more scalable and computationally efficient. In this paper, we use the integer program approach and focus more towards empirically bringing out fundamental tensions, trade-offs and difficulties involved in fair conference scheduling problem.}
\section{Experimental Evaluation}
\label{sec:experiments}
In this section, we present the baselines, introduce the metrics for evaluations and then show how our proposal compares with the baselines along the presented metrics.

\subsection{Baselines}\label{subsec:baselines}
We use the following baselines and empirically compare them with our approach from \cref{subsec:joint_opt} (further on referred to as {\bf FairConf}).
\begin{enumerate}[topsep=0pt,itemsep=-1ex,partopsep=1ex,parsep=1ex,font=\bfseries]
	\item Social Welfare Maximization ({\bf SWM}): 
	In this baseline, we just optimize the schedule for social welfare; i.e., $\Gamma^\text{SW}$ or $\argmax_\Gamma TEP(\Gamma)$ without any fairness consideration.
	\item Participant Fairness Maximization ({\bf PFair}):
	Here, we just optimize for participant fairness; i.e., minimize participant unfairness [$\argmin_\Gamma \Psi^\text{P}(\Gamma)$] as defined in \cref{eq:participant_fairness_obj}.
	\item Speaker Fairness Maximization ({\bf SFair}):
	Here, we just optimize for speaker fairness; i.e., minimize speaker unfairness [$\argmin_\Gamma \Psi^\text{S}(\Gamma)$] as defined in \cref{eq:speaker_fairness_obj}.
	\item Interest-Availability Matching ({\bf IAM}): 
	Here, we sort the talks in descending order of the overall interest scores received by them, i.e., $\sum_{p\in \mathcal{P}}V_p(t)$, and the slots in descending order of the overall availability scores received by them, i.e., $\sum_{p\in \mathcal{P}}A_p(s)$.
	Now, we assign the talk with the highest overall interest score to the slot with the highest overall availability score, the talk with the second highest overall interest score to the slot with the second highest overall availability score, and so on (with random tie-breaks).
	IAM is one of the naive alternatives when scheduling is done manually (as natural objectives like SWM usually need computing resources).
	It is also worth noting that, in the usual physical conference settings, both SWM and IAM give results which maximize social welfare, as we prove in claim \ref{claim:physical_conf}.
	%
	%
\end{enumerate}
\begin{claim}
In physical conference settings (i.e., all participants have identical ease of availability over all available slots $A_p(s)=A(s)$, $\forall s\in \mathcal{S},p\in \mathcal{P}$), {\bf IAM} yields a conference schedule which maximizes social welfare.
\label{claim:physical_conf}
\end{claim}
\begin{proof}
	As this is a special case of the scenario covered in \cref{lemma:V_or_A_same}, 
	refer to case (a) of \cref{lemma:V_or_A_same} for the proof.
\end{proof}
\begin{lemma}{\bf IAM} maximizes social welfare, if the participants are identical either in terms of their interests in the talks or in terms of their ease of availability over the available slots, or both.
	\label{lemma:V_or_A_same}
\end{lemma}
\begin{proof}
	There are three cases where we need to prove that IAM maximizes social welfare;
	{\bf (a)} if all participants have identical ease of availability over all available slots $A_p(s)=A(s)$, $\forall s\in \mathcal{S},p\in \mathcal{P}$ (this case is similar to physical conference settings where all participants gather at the same place, thus, have identical ease of availability);
	{\bf (b)} if all participants have identical interests over all talks $V_p(t)=V(t)$, $\forall t\in \mathcal{T},p\in \mathcal{P}$;
	{\bf (c)} if both (a) and (b) are true.
	
	We, first, reduce the SWM objectives in the following cases, and observe that they take a particular form where IAM gives solution.\\
	
	\noindent {\bf Case-(a):}
	Given $A_p(s)=A(s)$, $\forall s\in \mathcal{S},p\in \mathcal{P}$.
	From \cref{eq:TEP}:
	\[\text{SWM}\equiv \argmax_\Gamma\sum_{p\in\mathcal{P}}\sum_{t\in\mathcal{T}}V_p(t)\times A_p(\Gamma(t))\]
	\[\equiv \argmax_\Gamma\sum_{p\in\mathcal{P}}\sum_{t\in\mathcal{T}}V_p(t)\times A(\Gamma(t)) \equiv \argmax_\Gamma\sum_{t\in\mathcal{T}}\sum_{p\in\mathcal{P}}V_p(t)\times A(\Gamma(t))\]
	\[\equiv \argmax_\Gamma\sum_{t\in\mathcal{T}} A(\Gamma(t))\Big[\sum_{p\in\mathcal{P}}V_p(t)\Big] \equiv \argmax_\Gamma\sum_{t\in\mathcal{T}} A(\Gamma(t))\times \mathcal{V}(t)\]
	(where overall interest score $=\mathcal{V}(t)=\sum_{p\in\mathcal{P}}V_p(t)$)\\
	
	\noindent {\bf Case-(b):}
	Given $V_p(t)=V(t)$, $\forall t\in \mathcal{T},p\in \mathcal{P}$.
	From \cref{eq:TEP}:
	\[\text{SWM}\equiv \argmax_\Gamma\sum_{p\in\mathcal{P}}\sum_{t\in\mathcal{T}}V_p(t)\times A_p(\Gamma(t))\]
	\[\equiv \argmax_\Gamma\sum_{p\in\mathcal{P}}\sum_{t\in\mathcal{T}}V(t)\times A_p(\Gamma(t)) \equiv \argmax_\Gamma\sum_{t\in\mathcal{T}}\sum_{p\in\mathcal{P}}V(t)\times A_p(\Gamma(t))\]
	\[\equiv \argmax_\Gamma\sum_{t\in\mathcal{T}} V(t)\Big[\sum_{p\in\mathcal{P}}A_p(\Gamma(t))\Big] \equiv \argmax_\Gamma\sum_{t\in\mathcal{T}} V(t)\times \mathcal{A}(\Gamma(t))\]
	(where overall availability score $=\mathcal{A}(s)=\sum_{p\in\mathcal{P}}A_p(s)$)\\
	
	\noindent {\bf Case-(c):}
	Given $A_p(s)=A(s),V_p(t)=V(t)$, $\forall s\in \mathcal{S},\forall t\in \mathcal{T},p\in \mathcal{P}$.
	From \cref{eq:TEP}:
	\[\text{SWM}\equiv \argmax_\Gamma\sum_{p\in\mathcal{P}}\sum_{t\in\mathcal{T}}V_p(t)\times A_p(\Gamma(t))\]
	\[\equiv \argmax_\Gamma\sum_{p\in\mathcal{P}}\sum_{t\in\mathcal{T}}V(t)\times A(\Gamma(t)) \equiv \argmax_\Gamma\sum_{t\in\mathcal{T}}\sum_{p\in\mathcal{P}}V(t)\times A(\Gamma(t))\]
	\[\equiv \argmax_\Gamma\sum_{t\in\mathcal{T}}V(t)\times A(\Gamma(t))\]
	
	\noindent In all the three cases above, SWM reduces to a form where the terms are independent of individual participant interests and availability while depending only on the overall interest levels ($\mathcal{V}$ in case (a), and participant-independent $V$ in case (b) \& (c)) and overall availability ($\mathcal{A}$ in case (b), and participant-independent $A$ in case (a) \& (c)).
	Thus, to maximize the reduced objectives in all the cases, top values of overall availability $A$ or $\mathcal{A}$ need to be matched with top values of overall interests $\mathcal{V}$ or $V$ ---making it identical to IAM.
\end{proof}
Note that, when there are ties in scenarios mentioned in \cref{lemma:V_or_A_same}, IAM approach could also be made to give multiple solutions which maximize social welfare;
however, we use random tie breaks.
\subsection{Evaluation Metrics}\label{subsec:eval_metrics}
We use the following metrics to capture the performances from participant, speaker, and social welfare perspectives.

\subsubsection{\bf Participant-Side Metrics}
We measure the mean satisfaction of participants ($NCG_\text{mean}=\frac{\sum_{p\in \mathcal{P}}NCG_p}{\abs{\mathcal{P}}}$) as it is an indicator of how efficient is the schedule for participants.
We also measure the participant unfairness (as defined in \cref{eq:participant_unfairness_redn}: $\Psi^\text{P}=NCG_\text{max}-NCG_\text{min}$).
Lower values of $\Psi^\text{P}$ represent better participant fairness.

\subsubsection{\bf Speaker-Side Metrics}
\label{subsubsec:speaker_metrics}
Similar to participant-side metrics, on speaker-side also, we measure the mean satisfaction of speakers ($NEC_\text{mean}=\frac{\sum_{t\in \mathcal{T}}NEC_t}{\abs{\mathcal{T}}}$) as an indicator of efficiency of the schedule for speakers, and also the speaker unfairness (as defined in \cref{eq:speaker_unfairness_redn}: $\Psi^\text{S}=NEC_\text{max}-NEC_\text{min}$). Lower values of $\Psi^\text{S}$ represent better speaker fairness.

\subsubsection{\bf Social Welfare Metric}
While taking into account the participant fairness and speaker fairness, there could be some loss in social welfare or the total expected participation as illustrated in examples in \cref{sec:motivation}.
Thus, we also measure the social welfare achieved by the conference schedules obtained by our approach and the baseline approaches to gauge the loss of welfare in comparison to the maximum possible social welfare.
We use $TEP$ metric (as defined in \cref{eq:TEP}) for this.
\begin{figure*}[t!]
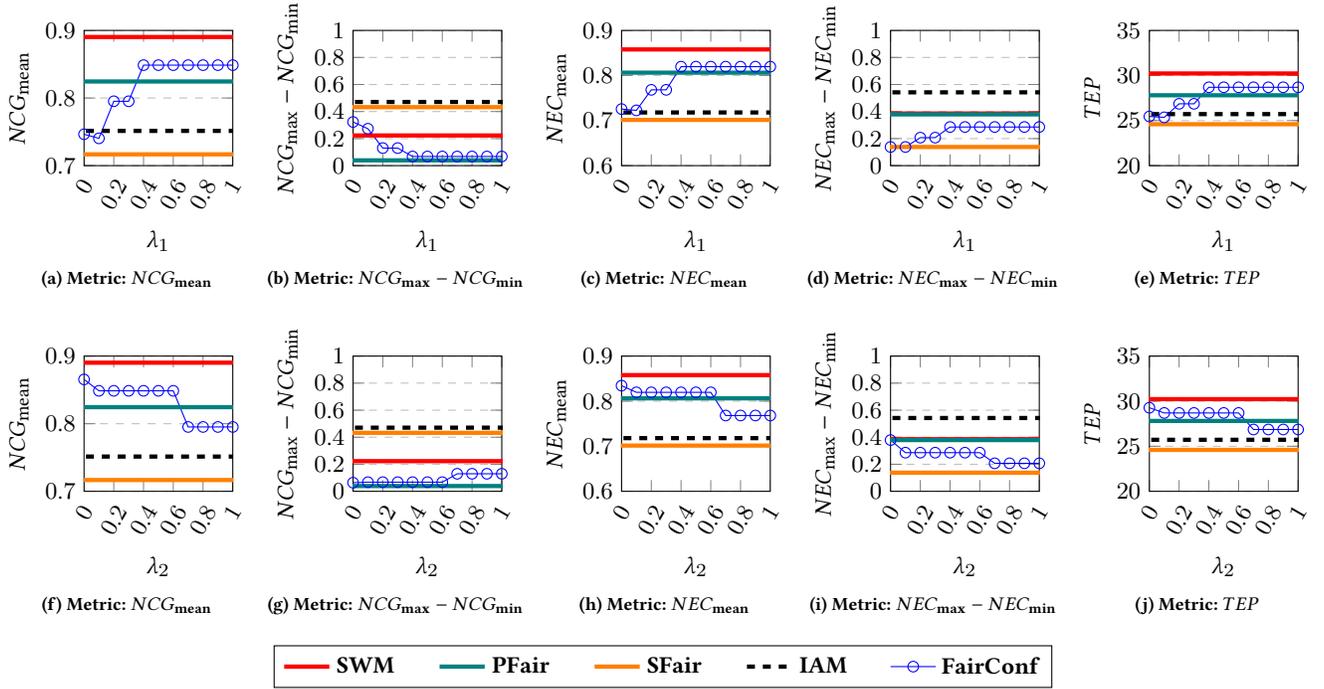

	\center{
		\subfloat[{\bf Metric: $NCG_\text{mean}$}]{\pgfplotsset{width=0.2\textwidth,height=0.19\textwidth,compat=1.9}
			\input{figures/syn_part_mu_lam1}\label{fig:syn_part_mu_lam1}}
		\hfil
		\subfloat[{\bf Metric: $NCG_\text{max}-NCG_\text{min}$}]{\pgfplotsset{width=0.2\textwidth,height=0.19\textwidth,compat=1.9}
			\input{figures/syn_part_unfair_lam1}\label{fig:syn_part_unfair_lam1}}
		\hfil
		\subfloat[{\bf Metric: $NEC_\text{mean}$}]{\pgfplotsset{width=0.2\textwidth,height=0.19\textwidth,compat=1.9}
			\input{figures/syn_spea_mu_lam1}\label{fig:syn_spea_mu_lam1}}
		\hfil
		\subfloat[{\bf Metric: $NEC_\text{max}-NEC_\text{min}$}]{\pgfplotsset{width=0.2\textwidth,height=0.19\textwidth,compat=1.9}
			\input{figures/syn_spea_unfair_lam1}\label{fig:syn_spea_unfair_lam1}}
		\hfil
		\subfloat[{\bf Metric: $TEP$}]{\pgfplotsset{width=0.2\textwidth,height=0.19\textwidth,compat=1.9}
			\input{figures/syn_TEP_lam1}\label{fig:syn_TEP_lam1}}
		\hfil
		\subfloat[{\bf Metric: $NCG_\text{mean}$}]{\pgfplotsset{width=0.2\textwidth,height=0.19\textwidth,compat=1.9}
			\input{figures/syn_part_mu_lam2}\label{fig:syn_part_mu_lam2}}
		\hfil
		\subfloat[{\bf Metric: $NCG_\text{max}-NCG_\text{min}$}]{\pgfplotsset{width=0.2\textwidth,height=0.19\textwidth,compat=1.9}
			\input{figures/syn_part_unfair_lam2}\label{fig:syn_part_unfair_lam2}}
		\hfil
		\subfloat[{\bf Metric: $NEC_\text{mean}$}]{\pgfplotsset{width=0.2\textwidth,height=0.19\textwidth,compat=1.9}
			\input{figures/syn_spea_mu_lam2}\label{fig:syn_spea_mu_lam2}}
		\hfil
		\subfloat[{\bf Metric: $NEC_\text{max}-NEC_\text{min}$}]{\pgfplotsset{width=0.2\textwidth,height=0.19\textwidth,compat=1.9}
			\input{figures/syn_spea_unfair_lam2}\label{fig:syn_spea_unfair_lam2}}
		\hfil
		\subfloat[{\bf Metric: $TEP$}]{\pgfplotsset{width=0.2\textwidth,height=0.19\textwidth,compat=1.9}
			\input{figures/syn_TEP_lam2}\label{fig:syn_TEP_lam2}}
		\vfil
		\subfloat{\pgfplotsset{width=.7\textwidth,compat=1.9}
			\begin{tikzpicture}
			\begin{customlegend}[legend entries={{\bf SWM},{\bf PFair},{\bf SFair},{\bf IAM}, {\bf FairConf}},legend columns=5,legend style={/tikz/every even column/.append style={column sep=0.5cm}}]
			\addlegendimage{red,mark=.,ultra thick,sharp plot}
			\addlegendimage{teal,mark=.,ultra thick,sharp plot}
			\addlegendimage{orange,mark=.,ultra thick,sharp plot}
			\addlegendimage{dashed,black,mark=.,ultra thick,sharp plot}
			\addlegendimage{blue,mark=o,sharp plot}
			\end{customlegend}
			\end{tikzpicture}}
	}\caption{Results on synthetic dataset with random interests and availability. For the plots in first row, $\lambda_2$ is fixed at $0.5$, and $\lambda_1$ is varied. For the plots in second row, $\lambda_1$ is fixed at $0.5$, and $\lambda_2$ is varied.}\label{fig:syn_random}
\end{figure*}
\subsection{Experimental Results}\label{subsec:syn_experiments}
First, in \cref{subsubsec:syn_random}, we show results on a synthetic dataset with random interest and availability scores. 
Then, in \cref{subsubsec:syn_segA,subsubsec:syn_segA_Imb,subsubsec:syn_segV,subsubsec:syn_segV_Imb}, we experiment on some special cases of synthetic datasets---with specific patterns in the data---to illustrate interesting nuances and difficulties involved in the scheduling problem.
Note that, we vary the hyperparameters $\lambda_1$ and $\lambda_2$ in between $0$ and $1$ in separate trials.
\subsubsection{\bf Random Interests and Availability:}\label{subsubsec:syn_random}
We take $\abs{\mathcal{P}}=10$, $\abs{\mathcal{T}}=10$, $\abs{\mathcal{P}}=10$, and generate the synthetic dataset where the slots represent non-overlapping equal-sized time intervals---not particularly required to be in any sequence---available for scheduling.
For this dataset, the interest scores and availability scores are sampled from a uniform random distribution in $[0,1]$; i.e., $V_p(t)\sim \text{Uniform}([0,1])$ and $A_p(s)\sim \text{Uniform}([0,1])$, $\forall p,t,s$.
We plot the results for this dataset in \cref{fig:syn_random}.
Note that, unlike FairConf, the baseline approaches do not have hyperparameters $\lambda_1,\lambda_2$;
thus, baseline results are just horizontal straight lines while FairConf's results vary with hyperparameter settings.
~\\{\bf Baseline Results:}
SWM achieves the highest expected participation $TEP$ (by definition it should), the highest mean participant satisfaction and mean speaker satisfaction (refer \cref{fig:syn_part_mu_lam1,fig:syn_spea_mu_lam1,fig:syn_TEP_lam1}) while performing poorly on participant and speaker fairness (\cref{fig:syn_part_unfair_lam1,fig:syn_spea_unfair_lam1}).
On the other hand, the naive IAM performs poorly in all the metrics.
As PFair optimizes only for participant fairness, it has the highest participant fairness (least unfairness in \cref{fig:syn_part_unfair_lam1}) while losing in all the other metrics.
The opposite happens with SFair;
it performs the best in speaker fairness (least unfairness in \cref{fig:syn_spea_unfair_lam1}) while losing in all other metrics as SFair optimizes only for speaker fairness.
~\\{\bf FairConf Results:}
Note that, for the plots in first row (\cref{fig:syn_part_mu_lam1} to \cref{fig:syn_TEP_lam1}), we fix $\lambda_2=0.5$ and vary $\lambda_1$ from $0$ to $1$;
for the plots in second row (\cref{fig:syn_part_mu_lam2} to \cref{fig:syn_TEP_lam2}), we fix $\lambda_1=0.5$ and vary $\lambda_2$ from $0$ to $1$.
The general trends observed in FairConf's results are: with increase in the weight for participant fairness ($\lambda_1$), FairConf achieves better participant fairness (decrease in participant unfairness in \cref{fig:syn_part_unfair_lam1}) but worse speaker fairness (increase in speaker unfairness in \cref{fig:syn_spea_unfair_lam1});
with increase in the weight for speaker fairness ($\lambda_2$), FairConf achieves better speaker fairness (decrease in speaker unfairness in \cref{fig:syn_spea_unfair_lam2}) but worse participant fairness (increase in participant unfairness in \cref{fig:syn_part_unfair_lam2}).
We find that FairConf with the setting of $\lambda_1=\lambda_2=0.5$ gives a balanced performance across all the metrics;
it performs good in both participant fairness (very small unfairness in \cref{fig:syn_part_unfair_lam1}-- also close to PFair) and speaker fairness (very small unfairness in \cref{fig:syn_spea_unfair_lam1}-- only little higher than SFair) while causing only marginal losses in mean participant satisfaction (\cref{fig:syn_part_mu_lam1}), mean speaker satisfaction (\cref{fig:syn_spea_mu_lam1}), and the social welfare (\cref{fig:syn_TEP_lam1}).

\begin{figure}
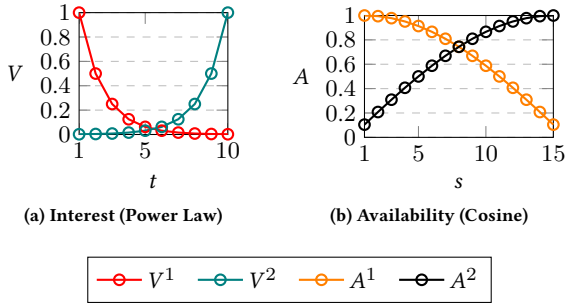

	\center{\subfloat[{\bf Interest (Power Law)}]{\pgfplotsset{width=0.2\textwidth,height=0.18\textwidth,compat=1.9}
			\input{figures/syn_V}\label{fig:syn_V}}
		\hfil
		\subfloat[{\bf Availability (Cosine)}]{\pgfplotsset{width=0.23\textwidth,height=0.18\textwidth,compat=1.9}
			\input{figures/syn_A}\label{fig:syn_A}}
		\vfil
		\subfloat{\pgfplotsset{width=.7\textwidth,compat=1.9}
			\begin{tikzpicture}
			\begin{customlegend}[legend entries={{\bf $V^1$},{\bf $V^2$},{\bf $A^1$},{\bf $A^2$}, {\bf FairConf}},legend columns=4,legend style={/tikz/every even column/.append style={column sep=0.2cm}}]
			\addlegendimage{red,mark=o,thick,sharp plot}
			\addlegendimage{teal,mark=o,thick,sharp plot}
			\addlegendimage{orange,mark=o,thick,sharp plot}
			\addlegendimage{black,mark=o,thick,sharp plot}
			\end{customlegend}
			\end{tikzpicture}}
	}\caption{Interest and Availability Patterns}\label{fig:syn_data}
\end{figure}
%
\begin{figure*}
	\center{\small
		\subfloat[{\bf Metric: $NCG_\text{mean}$}]{\pgfplotsset{width=0.2\textwidth,height=0.19\textwidth,compat=1.9}
			\input{figures/syn_segA_part_mu_lam1}\label{fig:syn_segA_part_mu_lam1}}
		\hfil
		\subfloat[{\bf Metric: $NCG_\text{max}-NCG_\text{min}$}]{\pgfplotsset{width=0.19\textwidth,height=0.19\textwidth,compat=1.9}
			\input{figures/syn_segA_part_unfair_lam1}\label{fig:syn_segA_part_unfair_lam1}}
		\hfil
		\subfloat[{\bf Metric: $NEC_\text{mean}$}]{\pgfplotsset{width=0.2\textwidth,height=0.19\textwidth,compat=1.9}
			\input{figures/syn_segA_spea_mu_lam1}\label{fig:syn_segA_spea_mu_lam1}}
		\hfil
		\subfloat[{\bf Metric: $NEC_\text{max}-NEC_\text{min}$}]{\pgfplotsset{width=0.2\textwidth,height=0.19\textwidth,compat=1.9}
			\input{figures/syn_segA_spea_unfair_lam1}\label{fig:syn_segA_spea_unfair_lam1}}
		\hfil
		\subfloat[{\bf Metric: $TEP$}]{\pgfplotsset{width=0.2\textwidth,height=0.19\textwidth,compat=1.9}
			\input{figures/syn_segA_TEP_lam1}\label{fig:syn_segA_TEP_lam1}}
		\hfil
		\subfloat[{\bf Metric: $NCG_\text{mean}$}]{\pgfplotsset{width=0.2\textwidth,height=0.19\textwidth,compat=1.9}
			\input{figures/syn_segA_part_mu_lam2}\label{fig:syn_segA_part_mu_lam2}}
		\hfil
		\subfloat[{\bf Metric: $NCG_\text{max}-NCG_\text{min}$}]{\pgfplotsset{width=0.19\textwidth,height=0.19\textwidth,compat=1.9}
			\input{figures/syn_segA_part_unfair_lam2}\label{fig:syn_segA_part_unfair_lam2}}
		\hfil
		\subfloat[{\bf Metric: $NEC_\text{mean}$}]{\pgfplotsset{width=0.2\textwidth,height=0.19\textwidth,compat=1.9}
			\input{figures/syn_segA_spea_mu_lam2}\label{fig:syn_segA_spea_mu_lam2}}
		\hfil
		\subfloat[{\bf Metric: $NEC_\text{max}-NEC_\text{min}$}]{\pgfplotsset{width=0.2\textwidth,height=0.19\textwidth,compat=1.9}
			\input{figures/syn_segA_spea_unfair_lam2}\label{fig:syn_segA_spea_unfair_lam2}}
		\hfil
		\subfloat[{\bf Metric: $TEP$}]{\pgfplotsset{width=0.2\textwidth,height=0.19\textwidth,compat=1.9}
			\input{figures/syn_segA_TEP_lam2}\label{fig:syn_segA_TEP_lam2}}
		\vfil
		\subfloat{\pgfplotsset{width=.7\textwidth,compat=1.9}
			\begin{tikzpicture}
			\begin{customlegend}[legend entries={{\bf SWM},{\bf PFair},{\bf SFair},{\bf IAM}, {\bf FairConf}},legend columns=5,legend style={/tikz/every even column/.append style={column sep=0.5cm}}]
			\addlegendimage{red,mark=.,ultra thick,sharp plot}
			\addlegendimage{teal,mark=.,ultra thick,sharp plot}
			\addlegendimage{orange,mark=.,ultra thick,sharp plot}
			\addlegendimage{dashed,black,mark=.,ultra thick,sharp plot}
			\addlegendimage{blue,mark=o,sharp plot}
			\end{customlegend}
			\end{tikzpicture}}
	}\caption{Results on data with balanced participant groups (indentical interests, segregated availability). For the plots in first row, $\lambda_2$ is fixed at $0.5$, and $\lambda_1$ is varied. For the plots in second row, $\lambda_1$ is fixed at $0.5$, and $\lambda_2$ is varied.}\label{fig:syn_segA}
\end{figure*}
%
\begin{figure*}
	\center{\small
		\subfloat[{\bf Metric: $NCG_\text{mean}$}]{\pgfplotsset{width=0.2\textwidth,height=0.19\textwidth,compat=1.9}
			\input{figures/syn_segA_Imb_part_mu_lam1}\label{fig:syn_segA_Imb_part_mu_lam1}}
		\hfil
		\subfloat[{\bf Metric: $NCG_\text{max}-NCG_\text{min}$}]{\pgfplotsset{width=0.2\textwidth,height=0.19\textwidth,compat=1.9}
			\input{figures/syn_segA_Imb_part_unfair_lam1}\label{fig:syn_segA_Imb_part_unfair_lam1}}
		\hfil
		\subfloat[{\bf Metric: $NEC_\text{mean}$}]{\pgfplotsset{width=0.2\textwidth,height=0.19\textwidth,compat=1.9}
			\input{figures/syn_segA_Imb_spea_mu_lam1}\label{fig:syn_segA_Imb_spea_mu_lam1}}
		\hfil
		\subfloat[{\bf Metric: $NEC_\text{max}-NEC_\text{min}$}]{\pgfplotsset{width=0.2\textwidth,height=0.19\textwidth,compat=1.9}
			\input{figures/syn_segA_Imb_spea_unfair_lam1}\label{fig:syn_segA_Imb_spea_unfair_lam1}}
		\hfil
		\subfloat[{\bf Metric: $TEP$}]{\pgfplotsset{width=0.2\textwidth,height=0.19\textwidth,compat=1.9}
			\input{figures/syn_segA_Imb_TEP_lam1}\label{fig:syn_segA_Imb_TEP_lam1}}
		\hfil
		\subfloat[{\bf Metric: $NCG_\text{mean}$}]{\pgfplotsset{width=0.2\textwidth,height=0.19\textwidth,compat=1.9}
			\input{figures/syn_segA_Imb_part_mu_lam2}\label{fig:syn_segA_Imb_part_mu_lam2}}
		\hfil
		\subfloat[{\bf Metric: $NCG_\text{max}-NCG_\text{min}$}]{\pgfplotsset{width=0.2\textwidth,height=0.19\textwidth,compat=1.9}
			\input{figures/syn_segA_Imb_part_unfair_lam2}\label{fig:syn_segA_Imb_part_unfair_lam2}}
		\hfil
		\subfloat[{\bf Metric: $NEC_\text{mean}$}]{\pgfplotsset{width=0.2\textwidth,height=0.19\textwidth,compat=1.9}
			\input{figures/syn_segA_Imb_spea_mu_lam2}\label{fig:syn_segA_Imb_spea_mu_lam2}}
		\hfil
		\subfloat[{\bf Metric: $NEC_\text{max}-NEC_\text{min}$}]{\pgfplotsset{width=0.2\textwidth,height=0.19\textwidth,compat=1.9}
			\input{figures/syn_segA_Imb_spea_unfair_lam2}\label{fig:syn_segA_Imb_spea_unfair_lam2}}
		\hfil
		\subfloat[{\bf Metric: $TEP$}]{\pgfplotsset{width=0.2\textwidth,height=0.19\textwidth,compat=1.9}
			\input{figures/syn_segA_Imb_TEP_lam2}\label{fig:syn_segA_Imb_TEP_lam2}}
		\vfil
		\subfloat{\pgfplotsset{width=.7\textwidth,compat=1.9}
			\begin{tikzpicture}
			\begin{customlegend}[legend entries={{\bf SWM},{\bf PFair},{\bf SFair},{\bf IAM}, {\bf FairConf}},legend columns=5,legend style={/tikz/every even column/.append style={column sep=0.5cm}}]
			\addlegendimage{red,mark=.,ultra thick,sharp plot}
			\addlegendimage{teal,mark=.,ultra thick,sharp plot}
			\addlegendimage{orange,mark=.,ultra thick,sharp plot}
			\addlegendimage{dashed,black,mark=.,ultra thick,sharp plot}
			\addlegendimage{blue,mark=o,sharp plot}
			\end{customlegend}
			\end{tikzpicture}}
	}\caption{Results on data with imbalanced participant groups (identical interests, segregated availability). For the plots in first row, $\lambda_2$ is fixed at $0.5$, and $\lambda_1$ is varied. For the plots in second row, $\lambda_1$ is fixed at $0.5$, and $\lambda_2$ is varied.}\label{fig:syn_segA_Imb}
\end{figure*}
%
\begin{figure*}
	\center{\small
		\subfloat[{\bf Metric: $NCG_\text{mean}$}]{\pgfplotsset{width=0.2\textwidth,height=0.19\textwidth,compat=1.9}
			\input{figures/syn_segV_part_mu_lam1}\label{fig:syn_segV_part_mu_lam1}}
		\hfil
		\subfloat[{\bf Metric: $NCG_\text{max}-NCG_\text{min}$}]{\pgfplotsset{width=0.2\textwidth,height=0.19\textwidth,compat=1.9}
			\input{figures/syn_segV_part_unfair_lam1}\label{fig:syn_segV_part_unfair_lam1}}
		\hfil
		\subfloat[{\bf Metric: $NEC_\text{mean}$}]{\pgfplotsset{width=0.2\textwidth,height=0.19\textwidth,compat=1.9}
			\input{figures/syn_segV_spea_mu_lam1}\label{fig:syn_segV_spea_mu_lam1}}
		\hfil
		\subfloat[{\bf Metric: $NEC_\text{max}-NEC_\text{min}$}]{\pgfplotsset{width=0.2\textwidth,height=0.19\textwidth,compat=1.9}
			\input{figures/syn_segV_spea_unfair_lam1}\label{fig:syn_segV_spea_unfair_lam1}}
		\hfil
		\subfloat[{\bf Metric: $TEP$}]{\pgfplotsset{width=0.2\textwidth,height=0.19\textwidth,compat=1.9}
			\input{figures/syn_segV_TEP_lam1}\label{fig:syn_segV_TEP_lam1}}
		\hfil
		\subfloat[{\bf Metric: $NCG_\text{mean}$}]{\pgfplotsset{width=0.2\textwidth,height=0.19\textwidth,compat=1.9}
			\input{figures/syn_segV_part_mu_lam2}\label{fig:syn_segV_part_mu_lam2}}
		\hfil
		\subfloat[{\bf Metric: $NCG_\text{max}-NCG_\text{min}$}]{\pgfplotsset{width=0.2\textwidth,height=0.19\textwidth,compat=1.9}
			\input{figures/syn_segV_part_unfair_lam2}\label{fig:syn_segV_part_unfair_lam2}}
		\hfil
		\subfloat[{\bf Metric: $NEC_\text{mean}$}]{\pgfplotsset{width=0.2\textwidth,height=0.19\textwidth,compat=1.9}
			\input{figures/syn_segV_spea_mu_lam2}\label{fig:syn_segV_spea_mu_lam2}}
		\hfil
		\subfloat[{\bf Metric: $NEC_\text{max}-NEC_\text{min}$}]{\pgfplotsset{width=0.2\textwidth,height=0.19\textwidth,compat=1.9}
			\input{figures/syn_segV_spea_unfair_lam2}\label{fig:syn_segV_spea_unfair_lam2}}
		\hfil
		\subfloat[{\bf Metric: $TEP$}]{\pgfplotsset{width=0.2\textwidth,height=0.19\textwidth,compat=1.9}
			\input{figures/syn_segV_TEP_lam2}\label{fig:syn_segV_TEP_lam2}}
		\vfil
		\subfloat{\pgfplotsset{width=.7\textwidth,compat=1.9}
			\begin{tikzpicture}
			\begin{customlegend}[legend entries={{\bf SWM},{\bf PFair},{\bf SFair},{\bf IAM}, {\bf FairConf}},legend columns=5,legend style={/tikz/every even column/.append style={column sep=0.5cm}}]
			\addlegendimage{red,mark=.,ultra thick,sharp plot}
			\addlegendimage{teal,mark=.,ultra thick,sharp plot}
			\addlegendimage{orange,mark=.,ultra thick,sharp plot}
			\addlegendimage{dashed,black,mark=.,ultra thick,sharp plot}
			\addlegendimage{blue,mark=o,sharp plot}
			\end{customlegend}
			\end{tikzpicture}}
	}\caption{Results on data with balanced participant groups (segregated interests, identical availability). For the plots in first row, $\lambda_2$ is fixed at $0.5$, and $\lambda_1$ is varied. For the plots in second row, $\lambda_1$ is fixed at $0.5$, and $\lambda_2$ is varied.}\label{fig:syn_segV}
\end{figure*}
%
\begin{figure*}
	\center{\small
		\subfloat[{\bf Metric: $NCG_\text{mean}$}]{\pgfplotsset{width=0.2\textwidth,height=0.19\textwidth,compat=1.9}
			\input{figures/syn_segV_Imb_part_mu_lam1}\label{fig:syn_segV_Imb_part_mu_lam1}}
		\hfil
		\subfloat[{\bf Metric: $NCG_\text{max}-NCG_\text{min}$}]{\pgfplotsset{width=0.2\textwidth,height=0.19\textwidth,compat=1.9}
			\input{figures/syn_segV_Imb_part_unfair_lam1}\label{fig:syn_segV_Imb_part_unfair_lam1}}
		\hfil
		\subfloat[{\bf Metric: $NEC_\text{mean}$}]{\pgfplotsset{width=0.2\textwidth,height=0.19\textwidth,compat=1.9}
			\input{figures/syn_segV_Imb_spea_mu_lam1}\label{fig:syn_segV_Imb_spea_mu_lam1}}
		\hfil
		\subfloat[{\bf Metric: $NEC_\text{max}-NEC_\text{min}$}]{\pgfplotsset{width=0.2\textwidth,height=0.19\textwidth,compat=1.9}
			\input{figures/syn_segV_Imb_spea_unfair_lam1}\label{fig:syn_segV_Imb_spea_unfair_lam1}}
		\hfil
		\subfloat[{\bf Metric: $TEP$}]{\pgfplotsset{width=0.2\textwidth,height=0.19\textwidth,compat=1.9}
			\input{figures/syn_segV_Imb_TEP_lam1}\label{fig:syn_segV_Imb_TEP_lam1}}
		\hfil
		\subfloat[{\bf Metric: $NCG_\text{mean}$}]{\pgfplotsset{width=0.2\textwidth,height=0.19\textwidth,compat=1.9}
			\input{figures/syn_segV_Imb_part_mu_lam2}\label{fig:syn_segV_Imb_part_mu_lam2}}
		\hfil
		\subfloat[{\bf Metric: $NCG_\text{max}-NCG_\text{min}$}]{\pgfplotsset{width=0.2\textwidth,height=0.19\textwidth,compat=1.9}
			\input{figures/syn_segV_Imb_part_unfair_lam2}\label{fig:syn_segV_Imb_part_unfair_lam2}}
		\hfil
		\subfloat[{\bf Metric: $NEC_\text{mean}$}]{\pgfplotsset{width=0.2\textwidth,height=0.19\textwidth,compat=1.9}
			\input{figures/syn_segV_Imb_spea_mu_lam2}\label{fig:syn_segV_Imb_spea_mu_lam2}}
		\hfil
		\subfloat[{\bf Metric: $NEC_\text{max}-NEC_\text{min}$}]{\pgfplotsset{width=0.2\textwidth,height=0.19\textwidth,compat=1.9}
			\input{figures/syn_segV_Imb_spea_unfair_lam2}\label{fig:syn_segV_Imb_spea_unfair_lam2}}
		\hfil
		\subfloat[{\bf Metric: $TEP$}]{\pgfplotsset{width=0.2\textwidth,height=0.19\textwidth,compat=1.9}
			\input{figures/syn_segV_Imb_TEP_lam2}\label{fig:syn_segV_Imb_TEP_lam2}}
		\vfil
		\subfloat{\pgfplotsset{width=.7\textwidth,compat=1.9}
			\begin{tikzpicture}
			\begin{customlegend}[legend entries={{\bf SWM},{\bf PFair},{\bf SFair},{\bf IAM}, {\bf FairConf}},legend columns=5,legend style={/tikz/every even column/.append style={column sep=0.5cm}}]
			\addlegendimage{red,mark=.,ultra thick,sharp plot}
			\addlegendimage{teal,mark=.,ultra thick,sharp plot}
			\addlegendimage{orange,mark=.,ultra thick,sharp plot}
			\addlegendimage{dashed,black,mark=.,ultra thick,sharp plot}
			\addlegendimage{blue,mark=o,sharp plot}
			\end{customlegend}
			\end{tikzpicture}}
	}\caption{Results on data with imbalanced participant groups (segregated interests, identical availability). For the plots in first row, $\lambda_2$ is fixed at $0.5$, and $\lambda_1$ is varied. For the plots in second row, $\lambda_1$ is fixed at $0.5$, and $\lambda_2$ is varied.}\label{fig:syn_segV_Imb}
\end{figure*}

\subsubsection{\bf Balanced Participant Groups (Indentical Interests, Segregated Availability):}\label{subsubsec:syn_segA}
Here, we take a dataset with $\abs{\mathcal{P}}=10$, $\abs{\mathcal{T}}=10$, $\abs{\mathcal{S}}=15$.
While all the participants have identical interest scores same as $V^1$ in \cref{fig:syn_V}, the first $5$ participants have availability scores as $A^1$ in \cref{fig:syn_A}, and the next $5$ have $A^2$ in \cref{fig:syn_A}.
We plot the results for this dataset in \cref{fig:syn_segA}.
~\\{\bf Baseline Results:}
Both SWM and IAM yield same social welfare (\cref{fig:syn_segA_TEP_lam1}), as it also follows from \cref{lemma:V_or_A_same} under the special condition of identical interest scores of the participants;
however, they could result in different optimal-welfare schedules;
here also we see different schedules given by SWM and IAM (difference in \cref{fig:syn_segA_part_unfair_lam1}).
As all participants have identical interests but segregated availability for two equal sized groups, there is huge scope for bringing in participant fairness by balancing the talks in favorable slots of both participant groups;
thus, PFair brings a huge improvement in participant fairness (\cref{fig:syn_segA_part_unfair_lam1}) in comparison to SWM.
However, there is no such scope for improvement in speaker fairness which is why SFair achieves only a marginal improvement (\cref{fig:syn_segA_spea_unfair_lam1}).
~\\{\bf FairConf Results:}
Again FairConf with $\lambda_1=\lambda_2=0.5$ performs very good in all the metrics here too.
While FairConf shows improvement in participant fairness (\cref{fig:syn_segA_part_unfair_lam1}) with increase in $\lambda_1$, it does not improve speaker fairness (\cref{fig:syn_segA_spea_unfair_lam2}) with increase in $\lambda_2$ due to very less scope for improving speaker fairness in this case.
However, it is worth noting that, just because there is less scope for improving speaker fairness, we should not just remove it from joint optimization by setting $\lambda_2=0$;
Setting $\lambda_2=0$ could adversely impact speaker fairness (see $\lambda_2=0$ point in \cref{fig:syn_segA_spea_unfair_lam2}).
Setting $\lambda_2=0$ can give the joint optimization an opportunity to further improve participant fairness (see $\lambda_2=0$ point in \cref{fig:syn_segA_part_unfair_lam2}) at the cost of losing speaker fairness.
Thus, it is important to keep reasonable non-zero weight $\lambda_2$ for speaker fairness---even in absence of any scope for improvement---as it can work both as optimizer and defender/preserver of speaker fairness.
\subsubsection{\bf Imbalanced Participant Groups (Identical Interests, Segregated Availability):}\label{subsubsec:syn_segA_Imb}
We use dataset same as previous, but with just one change:
here the first $7$ participants have availability scores as $A^1$ in \cref{fig:syn_A}, and the next $3$ have $A^2$ in \cref{fig:syn_A}.
We plot the results for this dataset in \cref{fig:syn_segA_Imb}.
~\\{\bf Baseline Results:}
In comparision to the case in \cref{subsubsec:syn_segA}, here SWM achieves higher participant satisfaction (compare SWM in \cref{fig:syn_segA_part_mu_lam1} and \cref{fig:syn_segA_Imb_part_mu_lam1}) and higher participation unfairness (compare SWM in \cref{fig:syn_segA_part_unfair_lam1} and \cref{fig:syn_segA_Imb_part_unfair_lam1}) too.
This is because, SWM can just assign the high interest talks to favorable slots of the majority participant group in order to maximize social welfare.
Thus, there is huge scope for improvement in participant fairness as the participant groups have segregated availability too.
That's why PFair brings a huge improvement in participant fairness (\cref{fig:syn_segA_Imb_part_unfair_lam1}).
However, just like the case in \cref{subsubsec:syn_segA}, there is no such scope to improve speaker fairness, thus, we see almost no improvement with SFair (\cref{fig:syn_segA_Imb_spea_unfair_lam1}).
~\\{\bf FairConf Results:}
Here also FairConf causes significant improvements in participant fairness (\cref{fig:syn_segA_Imb_part_unfair_lam1}) with increase in $\lambda_1$.
However, as there is no scope to improve speaker fairness, we see no improvement in speaker fairness (\cref{fig:syn_segA_Imb_spea_unfair_lam2}) with increase in $\lambda_2$.
Here also we see, it is not wise to set $\lambda_2=0$ just because there is no scope to improve speaker fairness;
here also, setting $\lambda_2=0$ adversely impacts speaker fairness and satisfaction (see $\lambda_2=0$ point in \cref{fig:syn_segA_Imb_spea_mu_lam2,fig:syn_segA_Imb_spea_unfair_lam2}).
Looking at the FairConf results with $\lambda_1=\lambda_2=0.5$, we can say that FairConf has handled the case of imbalanced availability segregation very well, and has produced very good results across all metrics.
\subsubsection{\bf Balanced Participant Groups (Segregated Interests, Identical Availability):}\label{subsubsec:syn_segV}
Here, we take a dataset with $\abs{\mathcal{P}}=10$, $\abs{\mathcal{T}}=10$, $\abs{\mathcal{S}}=15$.
While all the participants have identical availability scores same as $A^1$ in \cref{fig:syn_A}, the first $5$ participants have interests scores as $V^1$ in \cref{fig:syn_V}, and the next $5$ have $V^2$ in \cref{fig:syn_V}.
We plot the results for this dataset in \cref{fig:syn_segV}.
%
%

As the participant groups have segregated interest scores for the talks and identical availability scores over all slots, one would expect SWM to show large unfairness in the results.
However, the slopes of the chosen interest score pattern and availability score pattern also play a role in this.
Looking at the slopes of $V^1$, $V^2$ \cref{fig:syn_V} (power law slope), one can easily see that they decrease significantly faster than $A^1$ in \cref{fig:syn_A} (Cosine slope).
Thus, two talks with same overall interest scores, can be assigned two consecutive slots without causing too high disparity in individual participant satisfactions and individual speaker satisfactions, as the difference between the overall availability scores of two consecutive slots is not too large.
This is why here, we see SWM not only optimizes social welfare but also achieves high participant and speaker fairness (\cref{fig:syn_segV})---close to FairConf.

It is also worth noting that even though SWM provides a solution with the best speaker fairness (\cref{fig:syn_segV_spea_unfair_lam1}), SFair gives a different solution with same speaker fairness but with significantly poorer performances in other metrics.
Similarly, PFair also gives a different solution which improves participant fairness by a very small amount (\cref{fig:syn_segV_part_unfair_lam1}), but causes significant losses in other metrics.
This happens because both PFair and SFair are agnostic to the other fairness objective and welfare objective.
Thus, they may or may not result in the same schedule as SWM even if it is optimal.
\subsubsection{\bf Imbalanced Participant Groups (Segregated Interests, Identical Availability):}\label{subsubsec:syn_segV_Imb}
We use dataset same as previous, but with just one change:
here the first $7$ participants have interest scores as $V^1$ in \cref{fig:syn_V}, and the next $3$ have $V^2$ in \cref{fig:syn_V}.
We plot the results for this dataset in \cref{fig:syn_segV_Imb}.
%
%

In contrast to the case in \cref{subsubsec:syn_segV}, here, there is an imbalance in the interest segregation.
Thus, it provides SWM an opportunity to be biased towards the majority participant group and improve social welfare.
Thus, we see a higher participant unfairness (\cref{fig:syn_segV_Imb_part_unfair_lam1}).
FairConf significantly improves participant fairness while causing marginal loss in social welfare (\cref{fig:syn_segV_Imb_TEP_lam1}).

\subsubsection{\bf Results Summary:}\label{subsubsec:results_summary}
While SWM maximizes social welfare, it often results in high participant unfairness and speaker unfairness.
On the other hand, naive approach IAM also optimizes social welfare in special conditions (\cref{lemma:V_or_A_same}), but due to random tie breaks both SWM and IAM may not differentiate between two optimal-welfare schedules in terms of fairness.
Moreover, in absence of any explicit fairness consideration, both SWM and IAM often perform poorly in term fairness.
While PFair achieves maximum participant fairness, it often becomes unfair to the speakers, and also causes loss in mean speaker satisfaction;
the opposite happens in case of SFair.
Our joint optimization approach FairConf with similar weights for participant and speaker fairness, i.e., $\lambda_1=\lambda_2=0.5$ is found to be performing very good across all the metrics in all the tested cases (achieves good participant and speaker fairness with only marginal losses in social welfare, and overall participant satisfaction and overall speaker satisfaction).
\section{Related Work}\label{sec:related}
We briefly review related works in the following two directions.
\subsubsection*{\bf Job and Network Scheduling:}
The most commonly studied scheduling problem in computing research is on job or network scheduling.
In this problem, there are multiple agents (e.g., system processes, computing jobs, data packets, networked users or machines) who have shared access to common resource(s) (e.g., fixed number of processors, limited internet bandwidth), and the agents raise requests for the use the common resource(s) from time to time;
now the goal is to allocate the resource(s) to the agents in fair and optimal manner.
For example: 
fair-share scheduling for system processes \cite{kay1988fair,li2009efficient,lozi2016linux};
scheduling of packet transfers to ensure fair sharing of network channels \cite{vaidya2005distributed};
fair scheduling of computing jobs on computing clusters \cite{isard2009quincy,mahajan2019themis};
fair scheduling for devices in shared wireless charging systems \cite{fang2018fair};
fair scheduling of retrieval queries for databases \cite{harris2015fair}.
Our problem setup for fair conference scheduling is of very different form than typical job scheduling setup. 
While conference scheduling has two types of stakeholders---participants and speakers---whose functions and fairness requirements are different from each other, job scheduling problems are usually modelled for one type of stakeholders---the agents who use the shared resource.
\subsubsection*{\bf Meeting/Event Scheduling:}
The problem which is closely related to our conference scheduling problem is meeting or event scheduling where there are mulltiple agents with different availability in different time intervals, and the goal is to find an optimal schedule for meeting(s).
These problems have been explored for different types of optimality;
for example: finalizing schedules with minimal negotiations with agents \cite{sen1998formal}; schedules with optimal availability of the agents \cite{garrido1996multi,capek2008event}.
To solve these problems, methods like distributed constraint optimization \cite{maheswaran2004taking} have been proposed.
Even though there is a similarity between meeting scheduling and our conference scheduling, in meeting scheduling problems, utility/satisfaction is modelled only for the agents attending the event(s), and there is no consideration of satisfaction from the side of the event (i.e., no concept of speakers as individuals with self interests).
In addition to that most of these works consider only binary availability of the agents, and are often focussed more on optimizing efficiency.
While there has been work on strategy-proof scheduling \cite{ephrati1994non} and privacy-aware scheduling \cite{freuder2001privacy,lee2016probabilistic}, fairness has not been considered in such settings (except for \citet{baum2014scheduling} dealing with a very different setting).
In constrast, we allow non-binary availability of the agents, and care for both efficiency and fairness in our conference scheduling problem.
\section{Discussion}\label{sec:discussion}
In this work, we modeled a very timely and important problem of online conference scheduling with welfare and fairness concerns.
Apart from formally defining the fairness notions and objectives, we brought out several fundamental tensions among participant fairness, speaker fairness, and social welfare, 
%
and showed the benefits of our proposed joint optimization framework.
We believe that 
this work will lay the ground 
for further research (both theoretical and empirical) into such multi-stakeholder scheduling problems (going beyond the virtual conference scheduling).
Next, we present a set of open questions and possible future works.
\subsection*{Future Work}
(i) Even though, in this paper, we solved the proposed joint optimization problem by converting it into an integer linear program which works well for problems of small size, more research is needed in developing approaches which are more scalable and computationally efficient with provable guarantees. \\
%
(ii) We also limited ourselves with non-overlapping time slots; 
however, larger conferences often have overlapping and parallel sessions to accommodate their higher number of talks.
Accounting for such overlapping slots would require changes in the formulation of participant and speaker satisfactions as participants would have to choose which of the parallel session to attend, thereby also changing the expected 
audience in a particular talk. \\
(iii) We considered speaker satisfaction in terms of the expected crowd at their talks, however, in virtual conferences, the ease of availability of the speakers in the assigned time slot also plays a role in their satisfaction;
so one need to factor in  both expected crowd and ease of availability of the speaker to define an encompassing measure for speaker satisfaction. \\
(iv) Although we considered the participants and speakers to be separate agents in our model, a single agent could play the roles of both a participant and a speaker;
thus for such agents with dual roles, the participant satisfaction measure can be modified to exclude the time slot(s) in which they play the role of a speaker. \\
(v) Another important aspect which may be of great importance in conference scheduling is to consider 
 group fairness: that is to ensure fair satisfaction for timezone-specific groups of participants,
as well as domain-specific groups of speakers etc.
Group fairness notions form a long line of work in machine learning \cite{mehrabi2019survey}; these works can be explored and extended for this purpose.

\noindent {\bf Acknowledgements:}
G. K Patro is supported by a fellowship from Tata Consultancy Services.
This research was supported in part by a European Research Council (ERC) Advanced Grant for the project ``Foundations for Fair Social Computing", funded under the European Unions Horizon 2020 Framework Programme (grant agreement no. 789373).

\balance

\bibliographystyle{ACM-Reference-Format}
\bibliography{Main}
	
\end{document}